\documentclass[a4paper,12pt]{article}
\usepackage[english]{babel}
\usepackage{amsfonts,amsmath,amssymb,theorem}
\usepackage[unicode]{hyperref} 

\usepackage{ifthen}

\usepackage{color}
\usepackage{colortbl}
\usepackage{hhline}

\date{}

\newenvironment{proof}{\noindent \emph{Proof.}}{\hspace*{\fill} $\Box$\newline}

\newif\ifNoRemark
\def\addtheorem#1#2#3#4{
\ifthenelse{\equal{#2}{}}{}%
{\ifthenelse{\expandafter\isundefined\csname the#2\endcsname}{\newcounter{#2}}{}}
\newenvironment{#1}[1][\global\NoRemarktrue]
{\par\addvspace{2mm plus 0.5mm minus 0.2mm}\noindent 
{\bf #3}\ifthenelse{\equal{#2}{}}{}%
{\refstepcounter{#2}{\bf ~\csname the#2\endcsname}}%
{\bf \vphantom{##1}\ifNoRemark.\ \else\ (##1).\fi}\begingroup #4}%
{\endgroup\par\addvspace{1mm plus 0.5mm minus 0.2mm}\global\NoRemarkfalse}
\expandafter\newcommand\csname b#1\endcsname{\begin{#1}}
\expandafter\newcommand\csname e#1\endcsname{\end{#1}}
}

\addtheorem{proposition}{proCounter}{Proposition}{\sl}
\addtheorem{theorem}{thCounter}{Theorem}{\sl}
\addtheorem{lemma}{lemCounter}{Lemma}{\sl}
\addtheorem{corollary}{coCounter}{Corollary}{\sl}
\addtheorem{exam}{exCounter}{Example}{}
\addtheorem{remark}{}{Remark}{}

\newcommand{\word}[1]{\ensuremath{\mathbf{#1}}}

\begin{document}

\title{Non-existence of a ternary constant weight $(16,5,15;2048)$ diameter perfect code%
\thanks{The first author was partially supported by 
Grant 13-01-00463-a of 
the Russian Foundation for Basic Research 
and by Grant NSh–1939.2014.1 of President of
Russia for Leading Scientific Schools.
The third author was partially supported by the Finnish Cultural Foundation.}}
\author{%
Denis S. Krotov%
  \thanks{D.S.K.\ is with Sobolev Institute of Mathematics, Novosibirsk, Russia
  and with Novosibirsk State University, Novosibirsk, Russia},
Patric R. J. {\"O}sterg\aa{}rd%
  \thanks{P.R.J.\"O.\ is with the Department of Communications and Networking,
  Aalto University School of Electrical Engineering, P.O.\ Box 13000, 00076 Aalto,
  Finland}, 
Olli Pottonen%
  \thanks{O.P.\ is with the School of Mathematics and Physics, The University of Queensland,
Brisbane, Australia.}
}

\maketitle

\begin{abstract}
Ternary constant weight codes of length $n=2^m$, weight $n-1$,
cardinality $2^n$ and distance $5$ are known to exist 
for every $m$ for which there exists an APN permutation of order $2^m$,
that is, at least for all odd $m \geq 3$ and for $m=6$.
We show the non-existence of such codes for $m=4$ 
and prove that any codes with the parameters above are 
diameter perfect.
\end{abstract}


\section{Introduction}

A \emph{ternary constant-weight code} $C$ of length $n$, 
weight $w$, and minimum distance at least $d$, or an $(n,d,w;M)_3$ code,
is a set of $M$ words (codewords) over the alphabet $\{0,1,2\}$ 
with exactly $n-w$ 0s
such that every two distinct codewords differ in at least $d$ coordinates.
If $w=n$, then such a code is an unrestricted binary code over
the alphabet $\{1,2\}$. We consider the case $w=n-1$, which is, as we will see, 
also connected to binary codes.
For this reason, it is convenient to replace the alphabet 
$\{0,1,2\}$ by $\{*,0,1\}$ and consider words of length $n$ 
that contain exactly one $*$; the set of all such words is 
denoted by $X^n$, and the set of all binary words of
length $n$ is denoted by $F^n$.

The (Hamming) distance $d(\word{x},\word{y})$ 
between two words $\word{x},\word{y} \in X^n \cup F^n$ is the number of coordinates 
in which they differ. A word $\word{x} \in X^n$ will also be treated as a pair
of binary words differing 
in exactly one coordinate,
for example, $01{*}0 = \{0100,0110\}$.
One of these words is an even-weight word (that is, the number of 1s is even)
and the other is an odd-weight word;  
these are denoted by $e(\word{x})$ and $o(\word{x})$, respectively.
Moreover, we define $e(C)=\{e(\word{c}) \mid \word{c} \in C\}$ and 
$o(C)=\{o(\word{c}) \mid \word{c} \in C\}$.

We are here interested in the class of ternary constant weight codes 
with parameters $(n=2^m,5,n-1;2^{n-1}/n)_3$.
The first nontrivial example of such a code 
was an $(8,5,7;16)_3$ code constructed from a Jacobsthal matrix
\cite{Svanstr:PhD,OstSva2002}.
As proved in \cite{Kro:TernDiamPerf}
the existence of a so-called APN (almost perfect nonlinear) permutation 
$\{0,1\}^m\to \{0,1\}^m$ implies the existence of an
$(n=2^m,5,n-1;2^{n-1}/n)_3$ code $C$, 
where the codes $e(C)$ and $o(C)$
are cosets of binary extended Hamming codes.
APN permutations exist for every odd $m$ and for $m=6$ (see \cite{BDMW:2010}) 
and do not exist for $m=2,4$. The two exceptional values leave
the existence of $(4,5,3;2)_3$ and $(16,5,15;2048)_3$ codes open. 
Non-existence is trivial in the former case. The latter case
is the topic of the current note, where the following main result
will be obtained.

\begin{theorem}\label{th:16}
A $(16,5,15;2048)_3$ ternary constant weight code does not exist.
\end{theorem}

The paper is organized as follows.
In Section~\ref{s:computation} we describe a computational proof of
Theorem~\ref{th:16}.
In Section~\ref{s:Dperfect} we define 
the concept of a diameter perfect code and prove that 
a diameter perfect ternary constant weight code 
of length $n=2^m$, weight $n-1$
and minimum distance $5$ must have cardinality
$2^{n-1}/n$.

\begin{corollary}\label{cor:16}
There are no diameter perfect ternary constant weight codes of length $16$, 
weight $15$ and minimum distance $5$.
\end{corollary}

\section{Non-existence of $(16,5,15;2048)_3$ code}\label{s:computation}

Every ternary code of length $n$ and weight $n-1$ can be decomposed
into binary even and odd codes, as described in the Introduction, but 
even and odd binary codes cannot in general be combined to get a ternary code. 
However, for the parameters in question they can, as the next lemma shows. 
An $(n,M,d)$ code is a binary code of length $n$, size $M$ and minimum distance $d$.

\begin{lemma}\label{lem:comb1}
For any even code $C_0$ and odd code $C_1$, both $(n=2^m,2^{n-1}/n,4)$ codes, there is
an $(n=2^m,3,n-1;2^{n-1}/n)_3$ code $C$ such that $C_0 = e(C), C_1 = o(C)$.
\end{lemma}
\begin{proof}
It suffices to show that for any $\word{x} \in C_0$, there is unique $\word{y} \in C_1$
with $d(\word{x},\word{y}) = 1$. 
For any $\word{y} \in C_1$, the ball 
$B(\word{y}) = \{\word{z} \in F^n \mid d(\word{y},\word{z}) \leq 1\}$
contains $n$ even words. For distinct words $\word{y},\word{z} \in C_1$, we have
$B(\word{y}) \cap B(\word{z}) = \emptyset$ as $C_1$ has minimum distance $4$, 
so $\cup_{\word{y} \in C_1}B(\word{y})$ has cardinality $(2^{n-1}/n)\cdot n = 2^{n-1}$
and thereby contains all even words of $F^n$.

Therefore, for every even $\word{x} \in F^n$, there is unique $\word{y} \in C_1$ 
with $d(\word{x},\word{y}) = 1$. Obviously, this holds for $\word{x} \in C_0$. 
Finally, it is easy to check the claimed minimum distance.
\end{proof}

The $(n=2^m, 2^{n-1}/n, 4)$ codes are known as extended 1-perfect binary codes.
This construction only guarantees minimum distance 3. With an additional
restriction we get minimum distance 5.

\begin{lemma}\label{l:comb2}
Let $n=2^m$. A set $C\subset X^n$ is an $(n,5,n-1;2^{n-1}/n)_3$ code
if and only if the following two conditions holds.

{\rm (1)} Both $e(C)$ and $o(C)$ are binary $(n,2^{n-1}/n,4)$ codes.

{\rm (2)} If $\word{x}_1, \word{x}_2 \in e(C)$ and $\word{y}_1, \word{y}_2 \in o(C)$ satisfy
$d(\word{x}_1, \word{y}_1) = d(\word{x}_2, \word{y}_2) = 1$, $d(\word{x}_1, \word{x}_2) = 4$, then
$\word{x}_1 - \word{x}_2 \neq \word{y}_1 - \word{y}_2$.
\end{lemma}

\begin{proof}
Lemma \ref{lem:comb1} takes care of all parts but the minimum distance
and Condition (2). Two codewords in $C$ either have the $*$ in the same
coordinate or in different coordinates. In the former case, the words
in $e(C)$ obtained from the codewords will be at distance at least 5
from each other, that is, at least 6 since the distance is even. Then
the If part of Condition (2) is not fulfilled. In the latter case, 
$\word{x}_1-\word{y}_1 \neq \word{x}_2-\word{y}_2$, since these
differences give a word with a 1 in the coordinate of the $*$. These 
argument can be reversed, so equivalence holds.
\end{proof}

Next we consider symmetries of the space $F^n$. A permutation $\pi$ of
the set $\{1, 2, \ldots, n\}$ acts on words by permuting coordinates:
$\pi((c_1, \ldots, c_n)) = (c_{\pi^{-1}(1)}, \ldots,
c_{\pi^{-1}(n)})$.  A pair $(\pi, \word{x})$, $\word{x} \in F_n$ acts on $\word{c}$ as $(\pi,
\word{x})(\word{c}) = \pi(\word{c} + \word{x}) = \pi(\word{c}) + \pi(\word{x})$. 
These actions are distance-preserving, that is, they are \emph{isometries}. 
Two codes $C_1, C_2$ are \emph{equivalent} if $C_1 = (\pi, \word{x})(C_2)$ 
for some $\pi, \word{x}$.  A mapping $(\pi, x)$ that fulfills
$C = (\pi, \word{x})(C)$ is an \emph{automorphism}
of $C$. An automorphism of type $(\pi,\mathbf{0} = 00\cdots 0)$
is a \emph{symmetry} of $C$.

To determine existence of $(16,5,15;2048)_3$ codes,
we want to find, up to equivalence, all pairs $C_0, C_1$
such that $C_0 = e(C), C_1 = o(C)$ in
Lemma~\ref{l:comb2}. As starting point, we have the complete
classification of $(16,2048, 4)$ codes~\cite{OstPot:15}. Specifically, we
have an exhaustive list of 2165 equivalence class
representatives
\footnote{Available at arXiv:0806.2513 and {\tt http://www.iki.fi/opottone/codes}}.
Without loss of generality, we may fix $C_0$ to one of the codes
in this list. Then $C_0$ is an even code and
${\mathbf 0} \in C_0$.

However, for $C_1$ we were not able to utilize the classification, since
we are not free to choose arbitrary equivalence class representatives, and iterating 
over all possible codes is not feasible. Instead we search for $C_1$ with a direct approach.

Let $W_0$, $W_1$ consist of all even and odd weight words of $F^{16}$,
respectively. By the argument in the proof of
Lemma~\ref{lem:comb1}, for each $\word{x} \in W_0$ there is a unique $\word{y} \in
C_1$, with $d(\word{x}, \word{y}) = 1$. We shall search for a set $C_1$ which
satisfies this condition.

The search can be formulated as an instance of the \emph{exact cover problem}. 
In the exact cover problem we have sets $S$ and $U$ and a relation $R
\subset S \times U$. The task is to find a subset $Q \subset
U$ such that for any $x \in S$ there is unique $u \in
Q$ with $(x,u) \in R$. In our case $S = W_0$, $U = W_1$ 
and $R = \{ (\word{x}, \word{y}) \mid \word{x} \in W_0, 
\word{y} \in W_1, d(\word{x}, \word{y}) = 1\}$.

We use Condition (2) from Lemma~\ref{l:comb2} to prune the search.
Not only does this restrict the search and speed it up by many orders of magnitude, 
but it also guarantees that the results are relevant to the problem at hand.

For further improvement, we can attack this problem by first
solving some of its subproblems. Given an instance of the exact cover problem
with parameters $S$, $U$ and $R$, consider a set $S' \subset S$, and let 
$U' = \{u \in U \mid
(x,u) \in R \mbox{ for some } x \in S'\}$
and $R' = R|_{S' \times U'}$. 
Now we can first find all solutions $Q'$ of the instance with parameters 
$S', U', R'$, and then extend those in all possible
ways to solutions $Q \supset Q'$ of the original instance with parameters 
$S, U, R$. In particular, if the subproblem
has no solutions, neither has the original instance.

We emphasize that the approach does not
affect the solutions found, it only affects the performance of the
algorithm. We are free to choose $S'$ arbitrarily and may proceed
via a sequence of subproblems, for $S' \subset S'' \subset S''' \subset \cdots \subset S$.

We choose $S' = \{{\mathbf 0}\}$. Let $\word{e}^i$ be the 
word of weight 1 with a 1 in the $i$th coordinate. There are
clearly 16 solutions for the instance induced by $S'$, one for
each possible word $\word{e}^i$. Note, however, that it suffices to consider 
one solution from each orbit of the symmetry group of $C_0$. 
Specifically, if there is a permutation $\pi$ such that
$\pi(C_0) = C_0$ and $\pi(Q_1) = Q_2$ for two solutions $Q_1$ and $Q_2$,
then one of the solutions can be ignored. This observation was not
used for larger subproblems.

The subsequent subproblems are defined based on the solution of
$S'$, $\word{e}^i$. We let 
$S'' = \{\word{x} \in W_0 \mid d(\word{x}, \word{e}^i) = 3, x_1 \neq e^i_1 \}$,
where a subindex indicates the value of that particular coordinate.
Finally, $S''' = \{\word{x} \in W_0 \mid d(\word{x}, \word{e}^i) = 3 \}$.

A reader familiar with design theory may note that a solution 
for $S''$ is equivalent to a Steiner triple system of order $n-1$, and a
solution for $S'''$ is equivalent to a Steiner quadruple system of order $n$.

The following numerical values depend on the representatives of the equivalence
classes of codes; the codes from {\tt http://www.iki.fi/opottone/codes} were
used here. Out of the $2165$ codes $C_0$, only $102$ admit a solution to some
of the exact cover instances induced by $S''$. In no case
is there a solution for the instances induced by $S'''$. The computation took
about 24 CPU hours on a modern laptop. The exact cover instances were
solved with the \texttt{libexact} library~\cite{KasPot08}. The symmetries of codes
were computed with \texttt{bliss}~\cite{JunKas07}.

\section{Diameter perfect codes}\label{s:Dperfect}

This section is devoted to proving the following result.
\begin{theorem}\label{t:diperfect}
Let $n = 2^m$, $m \geq 3$. A diameter perfect $(n,5,n-1;M)_3$ code satisfies 
$M=2^{n-1}/n$.
\end{theorem}

First we need to define a diameter perfect code \cite{AhlAydKha},
the concept of which is based on the following 
generalized pigeonhole principle.

\begin{lemma}\label{l:diri}
Consider a set $S$ and a system $\mathcal{A}$ 
of subsets of $S$ of cardinality $M$
with the property that every element of $S$ belongs to a fixed 
(independent on the choice of the element) 
number of sets from $\mathcal{A}$.
Moreover, let $C\subset S$ be a set that intersects any set from $\mathcal{A}$
in at most $k$ elements. Then $$\frac{|C|}{|S|} \leq \frac{k}{M}.$$
\end{lemma}
\begin{proof}
Assume that each element of $S$ occurs in $m$ subsets in $\mathcal{A}$.
Double counting the occurrences of the elements in $S$ gives
$
|S| m = M |\mathcal{A}|.
$
Similar double counting for the elements in $C$ gives
$
|C| m \leq k |\mathcal{A}|.
$
The claim follows.
\end{proof}

Let $C\subset X^n$ be a code with minimum distance $d$  
and let $A\subset X^n$ be a set of diameter $d-1$ 
(that is, the mutual distances between its elements
do not exceed $d-1$). The set $X^n$ has a group
of isometries that acts transitively on its elements, that is, 
for any $\word{x}, \word{y} \in X^n$ there is an
isometry $\phi: X^n \rightarrow X^n$ such that $\phi(\word{x}) = \word{y}$.
By applying all isometries to $A$
we get a set system $\mathcal{A}$ of codes with diameter $d-1$,
and $C$ intersects each set of $\mathcal{A}$ in at most one codeword.
By Lemma~\ref{l:diri},
\begin{equation}
\label{eq:prod}
|C|\cdot|A|\leq |X^n|.
\end{equation}
If this bound holds with equality, then $C$ is known
as a \emph{diameter perfect code}. Then $A$ and $C$ have
maximal cardinalities among codes of diameter $d-1$ and
minimum distance $d$, respectively.

Next we present some auxiliary results.

\begin{lemma}\label{l:max2n}
A code $D\subset X^n$, $n \geq 16$, with minimum distance
$3$ and diameter at most $4$ 
(that is, only the distances $3$ and $4$ are allowed 
between two different words from $D$) has cardinality $|D|\leq n$.
\end{lemma}
\begin{proof}
Denote by
$D_i$ the set of words in $D$ with $*$ in the $i$th coordinate.
We shall first determine an upper bound on $|D_i|$.
We delete the $i$th coordinate and extend $D_i$ with a parity bit
to obtain an equidistant binary code with length $n$ and distance $4$. 
We call the code trivial if each coordinate has $|D_i|$ or $|D_i|-1$
equal values, and nontrivial otherwise. Deza \cite{D:73} showed that the
size of an nontrivial equidistant binary code with distance $2k$
is at most $k^2+k+2$ (which is 8 for $k=2$). The size of a trivial
equidistant binary code is easily seen to be at most $n/k$, 
that is, $n/2$ here. Combining these two cases, we get the bound 
$|D_i|\leq n/2$ for $n \geq 16$. 


If $|D_i|\leq 1$ for every $i=1,\ldots,n$, then obviously
$|D| = \sum_{i=1}^n |D_i| \leq n$. The rest of the proof, where we assume
that $|D_i|>1$ for at least one value of $i$, is divided
into three subcases.

(1) \emph{For some $i$, three of the words in $D_i$ have mutual distances
  $4$}.  Without loss of generality, $i=1$ and the three codewords are 
  $*110\cdots0$, $*00110\cdots0$ and $*0000110\cdots0$.
By looking at the distances to these three words, we see that
the only possible codewords of $D\setminus D_1$ have the
form $000\cdots 0*0\cdots$ and $100\cdots 0*0\cdots$. Moreover,
a consideration of the mutual distances between those $2(n-1)$ words reveals that
no more than two can be taken into a set with mutual distances at least $3$.  
So, $|D|\leq |D_1|+2 \leq n/2+2\leq n$ for $n \geq 16$.

(2) \emph{For some $i$, $D_i$ contains two words 
at distance $4$ from each other but no three such words}.
Now $|D_i|\leq 4$, since with size greater than 4, there
would be at least three words with the same parity in the binary
part and these three words would have mutual distances 4.

Without loss of generality, $i=1$ and $D_1$ contains the words 
$*00000\cdots0$ and $*11110\cdots0$.
Readily, every word in $D \setminus D_1$ must have the values 
$0,0,1,1$; $0,0,1,{*}$; or $0,1,1,{*}$, in an arbitrary order,
in coordinates $2$ to $5$ and cannot have 1s in coordinates $6$ to $n$.

It follows that a word in $D \setminus D_1$ is at distance $1$ 
from at least one binary word
having exactly two $1$s in coordinates $2$ to $5$ 
and 0s in coordinates $6$ to $n$. With two possible
values in the first coordinate, 
the number of such binary words is $2\binom{4}{2} = 12$.
Since no two words in $D \setminus D_1$ can be at distance
1 from the same word, we have $|D \setminus D_1|\leq 12$,
so $|D|= |D \setminus D_1| + |D_1| \leq 12 + 4 \leq n$.

(3) \emph{For every $i$, $D_i$ does not
contain words at mutual distance $4$}.
Now $|D_i| \leq 2$, and we may assume, without loss of generality, 
that $|D_1| = 2$ and $D_1=\{*0000\cdots0,*1110\cdots0\}$.
Let $D' = D \setminus (D_1 \cup D_2 \cup D_3 \cup D_4)$.
No word in $D'$ can have 1s in coordinates $5$ to $n$,
and no such word can have the values 000 and 111 in
coordinates 2 to 4. There are then 12 possible values in the
first four coordinates which come in six pairs of
complements. Since the codewords of such a pair cannot
belong to different sets $D_i$ (otherwise the mutual distance
would be $4+2=6$), we get that
$|D| = |D_1 \cup D_2 \cup D_3 \cup D_4| + |D'| \leq 2\cdot4 +6 \leq n$.
\end{proof}

\begin{proposition}\label{p:max2n}
If $B\subset X^n$, $n=2^m \geq 8$, is a set of diameter at 
most\/ $4$, then $|B|\leq n^2$. The inequality is tight, that is,
a set of diameter\/ $4$ and cardinality $n^2$ exists.
\end{proposition}
\begin{proof}
The set 
\begin{equation}
\label{eq:b}
B = \{ \word{y} \in X^n \mid d(\mathbf{0}, \word{y}) \leq 2 \}
\end{equation}
has diameter $4$ and cardinality $n^2$ and thereby 
proves the existence part.

For $n=8$, an $(8,5,7;128)_3$ code $C$ exists~\cite{Svanstr:PhD,OstSva2002}. 
With $B$ given by (\ref{eq:b}), $|C|\cdot |B| = |X^8|$, so $B$ has maximum 
cardinality in this case by (\ref{eq:prod}).
(Actually, this argument works for all $m$ for which 
$(n=2^m,5,n-1;2^{n-1}/n)_3$ codes are known to exist.)

For $n\geq 16$, let $F$ be an $(n=2^m,3,n-1;2^{n-1})_3$ code
\cite{Svanstr:PhD,Svanstr:1999,vLinTol:1999}.
Form $\mathcal{A}$
by applying all isometries to $F$, whereby the sizes of
the sets in $\mathcal{A}$ is $M=2^{n-1}$.
By Lemma~\ref{l:max2n}, an arbitrary set $B \subset X^n$ with diameter at most 4 
intersects each set of $\mathcal{A}$ in at most $k=n$ codewords. 
An application of Lemma~\ref{l:diri} with $C=B$ then gives that 
$|B|/(n2^{n-1}) \leq n/2^{n-1}$, that is, $|B| \leq n^2$.
\end{proof}

Now Theorem~\ref{t:diperfect} is obvious: 
use (\ref{eq:b}) as the set $A$ in the definition of diameter perfect codes.
Then $|C| = |X^n|/|A| = 2^{n-1}/n$.
In particular, a hypothetical diameter perfect code $C$ in $X^{16}$ would meet 
$|C|= 2^{15}/16 = 2048$.
By Theorem~\ref{th:16}, such a code does not exist, which proves Corollary~\ref{cor:16}.

Proposition~\ref{p:max2n} solves the so-called diametric problem
for the metric space $X^n$, $n=2^m\geq 8$ and diameter $4$. 
For the Johnson space and the $q$-ary Hamming space, the diametric problem
for an arbitrary diameter was completely solved in 
\cite{AhlKha:97:intersection, AhlKha:98:hamming}.

\end{document}